\documentclass[11pt,nofootinbib,notitlepage, tightenlines,reprint,pra]{revtex4-1}

\usepackage[normalem]{ulem}
\usepackage{amsthm}
\usepackage{amsmath}
\usepackage{bbm}
\usepackage{amssymb}
\usepackage{graphicx,epstopdf}
\usepackage{url}
\usepackage{dsfont}
\usepackage{float}
\usepackage{natbib}
\usepackage{bm}
\usepackage{ifthen}
\usepackage[usenames,dvipsnames]{color}
\usepackage{mathrsfs}
\usepackage[colorlinks=true,citecolor=blue,urlcolor=black]{hyperref}
\usepackage{float}
\newtheoremstyle{dotless}{}{}{\itshape}{}{\bfseries}{}{ }{}

\newcommand{\Tr}{\mathrm{ Tr }}
\newcommand{\tn}[1]{\textnormal{#1}}
\newcommand{\be}{\begin{equation}}
\newcommand{\ee}{\end{equation}}

\newcommand{\sket}[1]{{\ensuremath{\lvert#1\rangle}}}
\newcommand{\lket}[1]{{\ensuremath{\left\lvert#1\right\rangle}}}
\newcommand{\ket}[1]{\if@display\lket{#1}\else\sket{#1}\fi}

\newcommand{\sbra}[1]{{\ensuremath{\langle#1\rvert}}}
\newcommand{\lbra}[1]{{\ensuremath{\left\langle#1\right\rvert}}}
\newcommand{\bra}[1]{\if@display\lbra{#1}\else\sbra{#1}\fi}

\newcommand{\sbraket}[2]{{\ensuremath{\langle#1\rvert#2\rangle}}}
\newcommand{\lbraket}[2]{{\ensuremath{\left\langle#1\!\left\rvert\vphantom{#1}#2\right.\!\right\rangle}}}
\newcommand{\braket}[2]{\if@display\lbraket{#1}{#2}\else\sbraket{#1}{#2}\fi}

\newcommand{\sketbra}[2]{{\ensuremath{\lvert #1\rangle\!\langle #2\rvert}}}
\newcommand{\lketbra}[2]{{\ensuremath{\left\lvert #1\right\rangle\!\!\left\langle #2\right\rvert}}}
\newcommand{\ketbra}[2]{\if@display\lketbra{#1}{#2}\else\sketbra{#1}{#2}\fi}

\newcommand{\proj}[1]{\ketbra{#1}{#1}}

\DeclareMathOperator{\tr}{tr}
\newcommand{\strace}[2][@]{\ensuremath{\tr\ifthenelse{\equal{#1}{@}}{}{_{#1}}(#2)}}
\newcommand{\ltrace}[2][@]{\ensuremath{\tr\ifthenelse{\equal{#1}{@}}{}{_{#1}}\left(#2\right)}}
\newcommand{\ktrace}[2][@]{\ensuremath{\tr\ifthenelse{\equal{#1}{@}}{}{_{#1}}\left[#2\right]}}
\newcommand{\trace}[2][@]{\if@display\ltrace[#1]{#2}\else\strace[#1]{#2}\fi}

\theoremstyle{dotless}
\newtheorem{thm}{Theorem}

\theoremstyle{definition}

\usepackage{pdfpages}
\begin{document}
\title{Optimality of semiquantum nonlocality \\in the presence of high inconclusive rates }
\author{Charles Ci Wen \surname{Lim}}
\email{limc@ornl.gov}
\affiliation{ Quantum Information Science Group, Computational Sciences and Engineering Division, Oak Ridge National Laboratory, Oak Ridge, TN 37831-6418, USA.}

\begin{abstract}
Quantum nonlocality is a counterintuitive phenomenon that lies beyond the purview of causal influences.~Recently, Bell inequalities have been generalized to the case of quantum inputs, leading to a powerful family of semi-quantum Bell inequalities that are capable of detecting any entangled state.~Here, we focus on a different problem and investigate how the local-indistinguishability of quantum inputs and postselection may affect the requirements to detect semi-quantum nonlocality.~To this end, we consider a semi-quantum nonlocal game based on locally-indistinguishable qubit inputs and derive its postselected local and quantum bounds by using a novel connection to the local-distinguishability of quantum states.~Interestingly, we find that the postselected local bound is independent of the measurement efficiency and that the Bell violation increases with lower measurement efficiencies. 
\end{abstract} 
\maketitle

It is known that in quantum physics, there exist experiments in which correlations from measurements on entangled systems are at odds with our causal world views.~These correlations may be verified by using a family of statistical tests called Bell inequalities~\cite{Bell1964, Brunner2014RMP}, which are linear constraints on the set of correlations that are compatible with the principle of local causes~\cite{footnote0}.~In other words, if the correlations violate a Bell inequality, then the underlying physics must be nonlocal in nature.~Remarkably, apart from their foundational significance, Bell inequalities have also found practical applications in quantum cryptography and quantum state estimation~\cite{Acin2007,Pironio2010,Lim2013, Reichardt2013,Vazirani2014,Yang2014}.~For these reasons, quantum nonlocality is one of the most widely studied topics in quantum information science. 

Recently, a new paradigm called \emph{semi-quantum nonlocality} has emerged~\cite{Buscemi2012}, where observers use quantum inputs---instead of classical inputs---to specify their desired measurement settings.~Interestingly, by doing so, all entangled states are ``nonlocal'', in that for any entangled state there is always a semi-quantum Bell inequality with which violation is achieved.~This feature suggests that certain semi-quantum Bell inequalities are strong entanglement witnesses and thus could provide an unprecedented level of confidence in detecting entanglement using untrusted measurement devices.~For instance, see Ref.~\cite{Branciard2013} for a generic procedure that converts entanglement witnesses into measurement-device-independent entanglement witnesses, and Ref.~\cite{Xu2014} for the corresponding proof-of-principle experiment.~See also Refs.~\cite{Cavalcanti2013,Kocsis2015} for the connection to quantum steering~\cite{Wiseman2007}.

On a more general level, semi-quantum nonlocality admits the possibility of working with \emph{locally-indistinguishable} quantum inputs, a notion that is central to local quantum state discrimination~\cite{Peres1991, Bennett1999,Cosentino2013} and quantum data hiding~\cite{Terhal2001, DiVincenzo2002}.~For our purposes, we define such quantum inputs as quantum states that are indistinguishable at the level of local operations and shared randomness (LOSR)~\cite{Buscemi2012}, but distinguishable at the level of local quantum measurements assisted with shared entanglement (henceforth referred to as \emph{quantum strategies}).~In particular, our theoretical contribution recognizes that semi-quantum Bell inequalities using locally-indistinguishable quantum inputs can acquire the following two interesting properties:~(1) the ability to safely perform postselection and (2) the ability to achieve higher Bell violations with decreasing measurement efficiencies~\cite{footnote1}. 

The first property is based on the fact that postselection strategies due to the detection loophole~\cite{Eberhard1993, Berry2010, Branciard2011} are local filtering processes assisted with shared randomness.~Thus by the above definition, it is impossible for LOSR models to produce postselected correlations that are semi-quantum nonlocal---even if arbitrarily low measurement efficiencies are allowed.~The second property is due to the fact that the violation of a semi-quantum Bell inequality is directly related to the local-distinguishability of the quantum input states.~This connection implies that with a suitable choice of quantum inputs, it is possible to devise a semi-quantum Bell inequality whose optimal violation is achieved only if the measurement efficiencies fall below a certain threshold;~this is analogous to the optimal discrimination of non-orthogonal quantum states whereby inconclusive measurement elements are necessary~\cite{Barnett2009}. 
 
To illustrate the above properties, we analyze a semi-quantum Bell experiment inspired by the Clauser-Horne-Shimony-Holt~(CHSH) Bell experiment~\cite{CHSH1969}, and derive its postselected local bound and postselected maximum quantum bound for a given measurement efficiency.~To start with, let us first clarify the meaning of using quantum states to choose the measurements.~While it is clear what it means by classically choosing a measurement setting (e.g., turning a knob), in the case of quantum inputs, the notion of choosing a measurement is somewhat less obvious.~To sharpen this notion, we propose to think in terms of programmable quantum measurement (PQM) devices~\cite{Nielsen1997, Fiurasek2002}.~More specifically, a PQM device is a measurement device that accepts two quantum inputs, namely a quantum target system and a quantum program system, and then performs a measurement (determined by the state of the program system) on the target system.~In other words, one uses the state of the quantum program system to choose the desired measurement.~Therefore, we may view measurements in the semi-quantum nonlocality framework as \emph{untrusted} PQM devices whose measurements are purportedly determined by \emph{trusted}  quantum input systems, i.e., see Fig.~(\ref{fig1}).\\

{\textbf{Semi-quantum CHSH inequality.}~We consider a semi-quantum Bell experiment involving two distant observers, called Alice and Bob,~who each have a trusted local source of randomness, a trusted qubit preparation device, and an untrusted PQM device.~Note that the measurement-independence condition~\cite{{Barrett2011}} is thus implicitly assumed.~In each run of the experiment, Alice generates two random bits $\bar{x}=x_1x_2$ and prepares a program qubit using the following encoding scheme:~$\proj{\bar{x}}=H^{x_1}\proj{x_2}H^{x_1}$, where $\{\ket{x_2}\}_{x_2=0,1}$ is the computational basis and $H$ is the Hadamard matrix.~Then, she sends the prepared qubit to her PQM device for measurement and receives an outcome $a\in\{0,1,\varnothing\}$, where all inconclusive outcomes are assigned to $\varnothing$.~Likewise for Bob, we write $\bar{y}=y_1y_2$ and $b$ to denote his measurement choice and measurement outcome, respectively.~Furthermore, in what follows, we will refer to Alice's and Bob's qubit input systems as $\mathsf{A}$ and $\mathsf{B}$, respectively, and their corresponding quantum target systems as $\mathsf{A}'$ and $\mathsf{B}'$.

In the LOSR framework, untrusted measurements are modeled by a classical distribution $\{\Pr[\lambda]\}_\lambda$ and a corresponding set of conditional local positive-operator valued measure (POVM) operators,  $\{Q^\lambda_a\}_{a}$, $\{R^\lambda_b\}_{b}$, acting on systems $\mathsf{A}$ and $\mathsf{B}$, respectively.~Here, the classical variable $\lambda$ is a diagonal quantum state living in the Hilbert space of $\mathsf{A}'\otimes \mathsf{B}'$, and thus captures all the classical randomness that is pre-shared between the two measurement devices.~For a given pair of measurement choices,~$\omega_{\bar{x}}:=\proj{\bar{x}}$ and $\tau_{\bar{y}}:=\proj{\bar{y}}$, the conditional probability of observing outcomes $a$ and $b$ is given as 
\be \label{Eq1_locality}
\Pr\left[ a,b|\bar{x},\bar{y}\right]=\sum_\lambda  \Pr[\lambda]\Tr\left[{Q}_{a}^\lambda \omega_{\bar{x}}\right] \Tr \left[{R}_b^\lambda  \tau_{\bar{y}} \right],
\ee
which is synonymous to the locality condition assumed in standard Bell inequalities.~Also, we write $\{M_{a,b}\}_{a,b}$ to denote the effective two-qubit measurement acting on the qubit inputs, i.e., $M_{a,b} = \sum_{\lambda} \Pr[\lambda] {Q}_a^\lambda \otimes {R}_b^\lambda$.~Note that if $M_{a,b}$ is not separable for some $a,b$, then by definition the joint target state must be entangled, i.e., see Fig.~(\ref{fig1}).~Accordingly, a violation of Eq.~(\ref{Eq1_locality}) implies that the local PQM devices must share entanglement.

\begin{figure}[t]
\includegraphics[width=73mm]{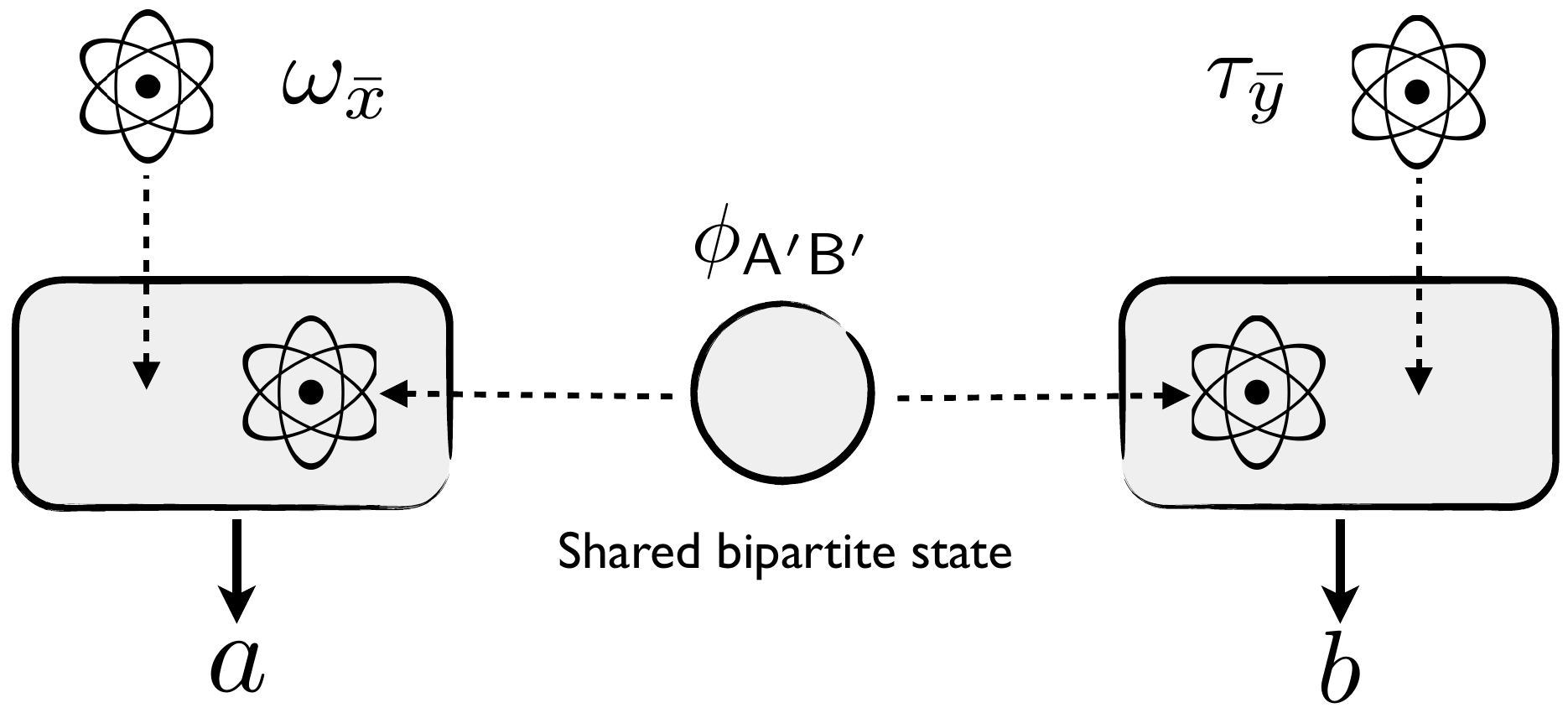}\caption{\textbf{Operational interpretation.}~Alice's and Bob's measurement choices are encoded into trusted qubit systems and then sent to their respective untrusted PQM devices. The PQM devices share a bipartite state (denoted by $\phi_{\mathsf{A'B'}}$) which may or may not be entangled. To test for entanglement, Alice and Bob compute Eq.~(\ref{Eq2_Ieq_LHS}): if the inequality is violated, they conclude $\phi_{\mathsf{A'B'}}$ is entangled, otherwise, the experiment is not conclusive.~It is useful to mention that like standard Bell experiments, the PQM devices and the source device are all part of the test.    }   
\label{fig1}
\end{figure}
Following standard arguments~\cite{Eberhard1993, Branciard2011}, we suppose $\sum_{a \not = \varnothing}\Pr\left[a |\bar{x} \right] = \gamma$, $\sum_{b \not = \varnothing}\Pr\left[b|\bar{y} \right]= \gamma$, and $\sum_{a,b \not = \varnothing}  \linebreak \Pr\left[a,b |\bar{x},\bar{y} \right] = \gamma^2$~for all measurements choices, where $\gamma \in (0,1]$ is the measurement efficiency.~With that, our postselected inequality reads
\be \label{Eq2_Ieq_LHS}
S(\gamma)=\frac{1}{4}\sum_{\bar{x},\bar{y}}(-1)^{f(\bar{x},\bar{y})}\frac{C(\bar{x},\bar{y})}{\gamma^2} \leq \beta(\gamma| \tn{LOSR}) ,
\ee
where $f(\bar{x},\bar{y}):=x_1 \wedge y_1 \oplus x_2 \oplus y_2$ is a balanced boolean function, and $C(\bar{x},\bar{y}):=\Pr[a=b|\bar{x},\bar{y}]-\Pr[a\not=b|\bar{x},\bar{y}]$ for $a,b\not= \varnothing$ is the conditional correlation function. Our goal is to derive the postselected local bound, $\beta(\gamma| \tn{LOSR})$, and to see how it scales with the measurement efficiency, $\gamma$.

For pedagogical reasons, we first discuss what happens when the inputs are classical. In this picture, our inequality can be seen as a symmetric extension of the CHSH inequality.~To see this connection, we note that the first bit of each party, $x_1,y_1$, determines his or her measurement setting, and the second bit, $x_2,y_2$, determines if he or she should flip the measurement outcome.~Indeed, it can be easily verified that Eq.~(\ref{Eq2_Ieq_LHS}) is an average of four CHSH inequalities conditioned on $x_2$ and $y_2$, therefore the local bound of our inequality assuming classical inputs is 2.~However, despite these similarities, there is a subtle difference between the CHSH inequality and Eq.~(\ref{Eq2_Ieq_LHS}) with respect to classical local models.~That is, a classical local model that outputs fixed correlated outcomes independently of the inputs would give a CHSH value of 2, whereas with Eq.~(\ref{Eq2_Ieq_LHS}) the Bell value is zero. This example illustrates that the additional randomness injected via $x_2$ and $y_2$ plays an interesting role in constraining the efficacy of certain classical local models. 

Moving on to quantum inputs, the measurement basis is now determined by the basis in which the program qubit is prepared, and the bit flip value is given by the eigenvector of that basis.~Notice that this encoding scheme is inspired by the celebrated quantum conjugate coding scheme used in quantum cryptography~\cite{Wiesner1983,BB84}.~The main advantage of this scheme is that the probability of learning each bit is upper bounded by $(1+1/\sqrt{2})/2 \approx 0.853$~\cite{Ambainis1999}. Based on these observations, we thus expect correlations generated by LOSR models to  be weakly correlated with Alice's and Bob's measurement choices. 

Recall that we want to derive the postselected local bound and the postselected maximum quantum bound for~Eq.~(\ref{Eq2_Ieq_LHS}). As mentioned earlier, the former is denoted by $\beta(\gamma|\tn{LOSR})$ and is defined as the maximization of $S(\gamma)$ over all LOSR measurements for a fixed measurement efficiency $\gamma$. At this point, it is useful to mention that all postselection strategies conceivable by LOSR models are automatically accounted for in the maximization.~That is, any postselection strategy employed by the underlying LOSR model must be captured by the local \emph{filtering} POVMs ${Q}_0^\lambda+{Q}_1^\lambda$ and ${R}_0^\lambda+{R}_1^\lambda$, which are also optimized as part of the maximization together with the distribution $\{\Pr[\lambda]\}_\lambda$.~Moving on, the postselected maximum quantum bound is denoted by $\beta(\gamma)$ and is defined as the maximization of $S(\gamma)$ over the set of quantum strategies, $\{\phi_{\mathsf{A'B'}},\{Q_a\}_a,\{R_b\}_b\}$.\\

{\textbf{Connection to quantum state discrimination.}} The above maximization problems can be solved by using a connection to the local-distinguishability of quantum inputs.~To illustrate this connection, we first note that the proposed semi-quantum CHSH experiment is equivalent to a guessing game in which the \emph{untrusted} local measurement devices have to guess the bit value $f(\bar{x},\bar{y})$ when given quantum inputs $\omega_{\bar{x}}\otimes \tau_{\bar{y}}$.~More precisely, the devices win the game if they output $a\oplus b=f(\bar{x},\bar{y})$ whenever $a,b\not= \varnothing$, i.e., the game is counted only for jointly conclusive events.~The conditional guessing probability can be written in terms of Eq.~(\ref{Eq2_Ieq_LHS}),
\be \label{Eq3_guessing}
G(\gamma):=\frac{\Pr\left[a\oplus b = f(\bar{x},\bar{y}) \right]}{\gamma^2}=\frac{1}{2}+\frac{S(\gamma)}{8},  \ee
where $S(\gamma)/8$ can be seen as the distinguishing advantage.~Then, it can be easily verified that 
\be  \label{Eq4_obj}
\Pr\left[a\oplus b = f(\bar{x},\bar{y}) \right]=\frac{\Tr\left[\rho_0\Pi_{a \oplus b =0}+\rho_1\Pi_{a \oplus b =1} \right]}{2},
\ee
where we used
\[
\rho_{0}=\frac{1}{8}\!\!\!\!\!\!\!\sum_{\substack{\bar{x},\bar{y}\\\tn{s.t.}f(\bar{x},\bar{y})=0}}\!\!\!\!\!\!\!\omega_{\bar{x}}\otimes \tau_{\bar{y}},\quad \, \rho_{1}=\frac{1}{8}\!\!\!\!\!\!\!\sum_{\substack{\bar{x},\bar{y}\\\tn{s.t.}f(\bar{x},\bar{y})=1}}\!\!\!\!\!\!\!\omega_{\bar{x}}\otimes\tau_{\bar{y}},\] and the measurement assignments $\Pi_{a\oplus b=0}=M_{0,0}+M_{1,1}$, $\Pi_{a\oplus b=1}=M_{0,1}+M_{0,1}$ and $\Pi_{\varnothing}=\mathds{1}-\Pi_{a\oplus b=0}-\Pi_{a\oplus b=1}$.

We may interpret Eq.~(\ref{Eq4_obj}) as follows. In each run of the experiment, the measurement devices are given a product quantum state $\omega_{\bar{x}}\otimes \tau_{\bar{y}}$ randomly chosen from one of the two sets of states, $\{\omega_{\bar{x}}\otimes \tau_{\bar{y}}: f(\bar{x},\bar{y})=0\}$ and $\{\omega_{\bar{x}}\otimes \tau_{\bar{y}}: f(\bar{x},\bar{y})=1\}$, and the devices have to guess which set the given state is drawn from.~In other words, the local devices have to collectively guess the global identity $f(\bar{x},\bar{y})$ of $\omega_{\bar{x}}\otimes \tau_{\bar{y}}$ using whatever resources they are given with.~Indeed, the figure of merit in this case is exactly given by Eq.~(\ref{Eq3_guessing}), which is the conditional guessing probability uniformly averaged over all product states.~Using the Born's rule and the linearity of the trace operator, this guessing game can be simplified to the local-distinguishability of two non-orthogonal mixed states $\rho_0$ and $\rho_1$ assuming a fixed conclusive rate of $\gamma^2$. Therefore, the maximization Eq.~(\ref{Eq2_Ieq_LHS}) is equivalent to the maximization of Eq.~(\ref{Eq4_obj}) (up to the constant normalization factor of $1/\gamma^2$). 

The advantage of local-distinguishability games is that they can be analytically solved through semidefinite programming~\cite{Cosentino2013}, a form of convex optimization that maximizes a linear function over the intersection of a semidefinite cone and an affine plane~\cite{SDP1996}.~For brevity, we present only the primal programs and defer the corresponding dual programs and optimal solutions to the Supplementary Material.~The primal program for computing the maximum quantum guessing probability assuming a fixed $\gamma^2 \in (0,1]$ is given by
\begin{eqnarray*}\nonumber
\texttt{maximize}&:& \frac{1}{2}\Tr \left[\rho_0  \Pi_{a\oplus b=0}  + \rho_1 \Pi_{a\oplus b=1} \right] \\ \nonumber
\texttt{subject to}&:&  \Pi_{a\oplus b=0}+ \Pi_{a\oplus b=1}+ \Pi_{\varnothing} = \mathds{1}_{\mathsf{A} \otimes \mathsf{B}},\\  \nonumber
&& \Tr\left[ (\omega_{\bar{x}}\otimes \tau_{\bar{y}}\Pi_{\varnothing}\right] =1-\gamma^2,\quad \forall~\bar{x},\bar{y} \\
&& \Pi_i \succeq 0,\quad i=0,1,\varnothing,
\end{eqnarray*}
and the optimal values are found to be
\be \label{Eq5_quantumbound}
\tn{max}\, G(\gamma)  =\left\{\begin{array}{lll}
             \frac{1}{2}\left(1+\frac{1}{\gamma^22\sqrt{2}}\right)\quad &\tn{if}&\quad \gamma>\frac{1}{\sqrt{2}}\\
             \frac{1}{2}+\frac{1}{2\sqrt{2}} \quad &\tn{if}&\quad \gamma\leq\frac{1}{\sqrt{2}}
            \end{array}\right..
\ee
Here, an important remark is in order.~These optimal values are obtained over the whole set of two-qubit POVMs acting on $\mathsf{A}\otimes\mathsf{B}$, which is larger than the set of quantum strategies, i.e., $M_{a,b}=\Tr_{\mathsf{A}'\mathsf{B}'}\left[ \phi_{\mathsf{A'B'}}(Q_a \otimes R_b)\right]$.~Thus strictly speaking, Eq.~(\ref{Eq5_quantumbound}) is an upper bound on the maximum quantum bound, i.e., $\beta(\gamma)\leq 2\sqrt{2}$ for $0<\gamma \leq 1/\sqrt{2}$ and $\beta(\gamma)\leq \sqrt{2}/\gamma^2$ for $1/\sqrt{2}<\gamma \leq 1$. However, as we will see later, this upper bound is tight for $\gamma \in (0,1/2]$.

Similarly, the maximization for LOSR measurements is based on a circuitous method, which nevertheless also leads to a tight upper bound on $G(\gamma|\tn{LOSR})$.~More precisely, we optimize over all measurements compatible with the positive partial transpose (PPT) condition~\cite{HorodeckiPPT} instead of LOSR measurements.~The reason is that PPT measurements admit a much simpler characterization and can be formulated as linear constraints in the semidefinite programs, i.e., we only need to add $ \Pi_i^{T_\mathsf{B}} \succeq 0$, for $i=0,1,\varnothing$, where $T_\mathsf{B}$ means the partial transpose with respect to Bob's measurements.~Moreover, we use the fact that PPT and separable measurements are equivalent at the level of two-qubit positive operators~\cite{HorodeckiPPT}.~Therefore, the optimal bound for PPT measurements is an upper bound on that of LOSR measurements, i.e., $\beta(\gamma|\tn{LOSR}) \leq \beta(\gamma|\tn{Sep}) = \beta(\gamma|\tn{PPT})$.~The optimal value for PPT models is found to be independent of $\gamma$, 
\be \label{Eq6_PPTbound}
\tn{max}\, G( \cdot|\tn{PPT}) =\frac{1}{2}+\frac{1}{4\sqrt{2}}.
\ee 
Interestingly, it turns out that the optimal measurements are given by LOSR measurements.~To show this, suppose the qubit inputs are given by $\omega_{0x_2}= (\mathds{1}_{\mathsf{A}}+(-1)^{x_2}\mathbbm{X})/2$, $\omega_{1x_2}= (\mathds{1}_{\mathsf{A}}+(-1)^{x_2}\mathbbm{Y})/2$, and $\tau_{y_1y_2}= (\mathds{1}_{\mathsf{B}}+(-1)^{y_2}(\mathbbm{X}+(-1)^{y_1}\mathbbm{Y})/\sqrt{2})/2$, where $\mathbbm{X}$ and $\mathbbm{Y}$ are Pauli matrices~\cite{footnote2}.~Then, it can be verified that the joint input states are jointly diagonal in the standard Bell basis:

 \be 
 \label{Eq7_diag_basis} \rho_0=   \begin{pmatrix}
    \alpha^+ & 0 & 0 & 0 \\
    0 &\alpha^- & 0 & 0 \\
    0 & 0& \frac{1}{4}& 0 \\
    0 & 0 & 0 & \frac{1}{4} 
  \end{pmatrix}, \quad \rho_1=   \begin{pmatrix}
    \alpha^- & 0 & 0 & 0 \\
    0 & \alpha^+ & 0 & 0 \\
    0 & 0& \frac{1}{4}& 0 \\
    0 & 0 & 0 & \frac{1}{4}
  \end{pmatrix},
 \ee where the eigenvalues are $\alpha^\pm:=(1\pm1/\sqrt{2})/4$, and the corresponding eigenvectors are ordered as: $\ket{\Psi^+}$, $\ket{\Psi^-}$, $\ket{\Phi^+}$ and $\ket{\Phi^-}$~\cite{footnote3}.~For example, we have $\bra{\Psi^+} \rho_0 \ket{\Psi^+}=\bra{\Psi^-} \rho_1 \ket{\Psi^-}=\alpha^+$.~A simple LOSR measurement that achieves Eq.~(\ref{Eq6_PPTbound}) is one that uses only local measurements, i.e., no shared randomness is needed.~More specifically, the strategy is $Q_a=\gamma(\mathds{1}_{\mathsf{A}}+(-1)^{a}\mathbbm{X})/2$, $R_b=\gamma(\mathds{1}_{\mathsf{B}}+(-1)^{b}\mathbbm{X})/2$ for $a,b=0,1$, and $Q_\varnothing=(1-\gamma)\mathds{1}_{\mathsf{A}}$, $R_\varnothing=(1-\gamma)\mathds{1}_{\mathsf{B}}$ for the inconclusive outcomes.~That is, each measurement device with probability $\gamma$ measures in the $\mathbbm{X}$ basis, and with probability $1-\gamma$ outputs $\varnothing$ without measurement.~Another strategy is to measure in the $\mathbbm{Y}$ basis instead of $\mathbbm{X}$, or to use a combination of these two strategies assisted with shared randomness. \newline

\begin{figure}[t]
  \includegraphics[width=84mm]{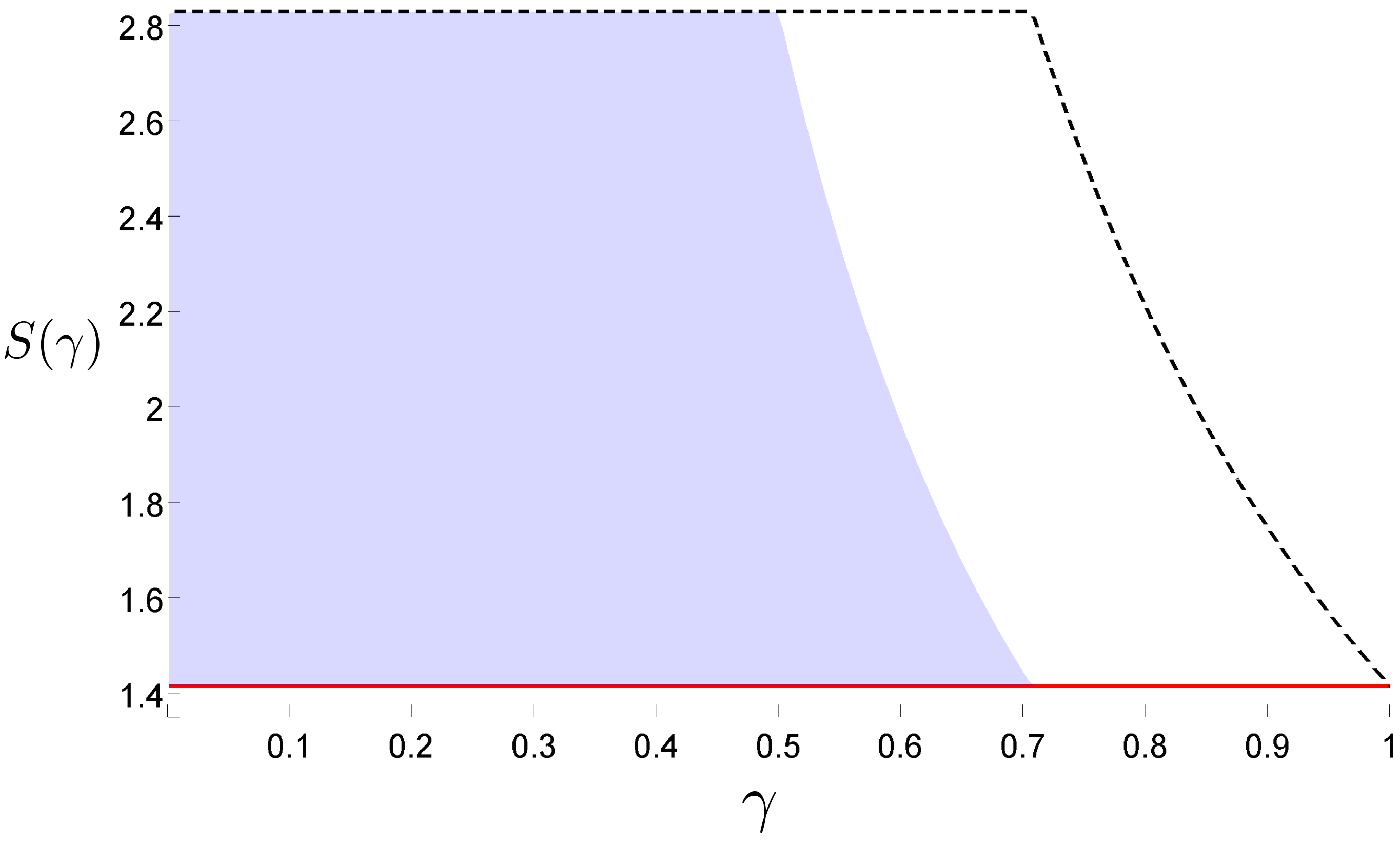}
  \caption{\textbf{Quantum violation vs efficiency.}~The vertical axis is the postselected Bell value $S(\gamma)$ and the horizontal axis is the measurement efficiency, $\gamma\in (0,1]$. The (black) dashed line is given by the maximum quantum bound, Eq.~(\ref{Eq5_quantumbound}), which is obtained using general two-qubit measurements. The (red) solid line is the postselected local bound, Eq.~(\ref{Eq6_PPTbound}). The (blue) shaded area is due to the pretty good quantum strategy.  }\label{fig2}   
\end{figure} 

{\textbf{A pretty good quantum strategy.}}~As mentioned above, the optimal solutions to Eq.~(\ref{Eq5_quantumbound}) are given in terms of two-qubit POVMs and thus do not provide a clear exposition on the optimal quantum strategy (i.e.,~the optimal entangled bipartite state and local PQMs) needed to achieve the maximum quantum bound.~To this end, we provide an explicit quantum strategy that reaches the upper bound in the region of $0<\gamma<1/2$, i.e., see the shaded area in Fig.~(\ref{fig2}).~Again, we refer to the aforementioned encoding scheme, i.e., Eq.~(\ref{Eq7_diag_basis}).~The optimal joint target system is a  two-qubit maximally entangled state, $\phi_{\mathsf{A'B'}}=\proj{\Psi^+}$, and the optimal PQMs are
\begin{eqnarray*}
Q_0&=&\gamma_1\proj{\Psi^+}+\gamma_2(\proj{00}+\proj{11})\\
Q_1&=&\gamma_1\proj{\Psi^-}+\gamma_2(\proj{00}+\proj{11})\\
Q_\varnothing&=&\mathds{1}_{\mathsf{A} \otimes \mathsf{B}} - Q_0-Q_1,
\end{eqnarray*} 
and likewise $R_i=Q_i$ for $i=0,1,\varnothing$, where $\gamma_1=\min\{2\gamma,1\}$ and $\gamma_2=\max\{\gamma-1/2,0\}$.~Note that the PQMs are inefficient Bell-state measurements (BSMs), i.e., they can only discriminate between $\ket{\Psi^+}$ and $\ket{\Psi^-}$.~The Bell values using these states and measurements are $S(\gamma)=2\sqrt{2}$ for $0<\gamma \leq 1/2$ and $S(\gamma)=1/(\gamma^2 \sqrt{2})$ for $1/2< \gamma <1/\sqrt{2}$.~We remark that this quantum strategy is however sub-optimal when it comes to detecting weakly entangled states.~For instance, in the case of two-qubit Werner states~\cite{Werner1989}, $\phi_{\mathsf{A'B'}}=F\proj{\Psi^-}+(1-F)\mathds{1}_{\mathsf{A'}\otimes \mathsf{B'}}/4$, it can be shown that violation is obtained only for $F>1/2$; note that these Werner states are separable for $F \leq 1/3$. On the other hand, we have upper bounds on the achievable Bell violations for $F>1/3$, which suggest that Eq.~(\ref{Eq2_Ieq_LHS}) might be able to detect all entangled two-qubit Werner states; see Supplementary Material.  \\

{\textbf{Discussion.}}~A way to interpret our result is to examine the optimality conditions for discriminating $\rho_0$ and $\rho_1$.~To begin with, we remind that these mixed states share the same support and can be simultaneously diagonalized in the Bell basis.~The first point implies that unambiguous state discrimination~\cite{Barnett2009} is not possible, thus the best measurement scheme, for our purpose, is probabilistic minimum-error state discrimination~\cite{Bagan2012,Herzog2012}.~From the optimality conditions of this scheme, it can be easily verified that the maximum success probability for which $\rho_0$ and $\rho_1$ are optimally discriminated is 1/2, which is indeed the value given in Eq.~(\ref{Eq5_quantumbound}).~This also explains the trend seen in Fig.~(\ref{fig2}) wherein higher Bell violations are achieved with higher inconclusive rates/lower measurement efficiencies.

From the second point, it is clear that the optimal measurement that discriminates between $\rho_0$ and $\rho_1$ must consists of entangled POVMs: the positive and negative eigenspaces of~$\rho_0-\rho_1$ are maximally entangled subspaces.~This means that no LOSR measurement (or more generally, separable measurement) can coherently access these entangled eigenspaces.~Crucially, this limitation also holds in the presence of inconclusive outcomes, i.e., entanglement cannot be created using local operations and classical communication (local filtering with shared randomness in our case).\\

{\textbf{Conclusion.}}~In the above, we have provided a semi-quantum Bell experiment that safely allows for postselection and is defined by a loss-independent local bound that is violated only in the region of imperfect measurement efficiencies.~On the conceptual level, our result suggests that semi-quantum nonlocality is much more powerful than previously recognized.~For instance, Eq.~(\ref{Eq2_Ieq_LHS}) does not require the so-called \emph{fair-sampling condition}~\cite{CHSH1969,Berry2010}, which is typically assumed in standard Bell experiments involving postselection to ensure that the conclusive/detected events are representative of the underlying quantum system.~Most interestingly, Fig.~(\ref{fig2}) shows that in order to (optimally) violate Eq.~(\ref{Eq2_Ieq_LHS}), it is necessary to use highly inefficient measurements, which up to the best of our knowledge, is the first time that such a trend has been found.~Furthermore, the maximal quantum violation $2\sqrt{2}$ can be achieved for a continuum of measurement efficiencies, i.e., $\gamma \in (0, 1/2]$, unlike standard Bell inequalities which can only reach their maximum violations in the limit of perfect measurement efficiency. 

Finally, we remark that on the practical side, our inequality provides a semi-device-independent method for testing entanglement in~\emph{detected quantum systems}.~That is, as mentioned above, the inequality allows one to restrict the analysis to detected events without assuming the fair-sampling condition.~For example, this application could be useful for entanglement-based experiments suffering from high detection losses, e.g., those based on practical entangled photon-pair sources~\cite{Valentina2015}. \newline

{\textbf{Acknowledgements.}}~We thank J.-D. Bancal, A. Martin,~V. Scarani, D. Rosset, N. Gisin, H.-K. Lo, R. Thew, B.~Qi, W. Grice, N. Johnston and A. Cosentino  for helpful discussions.~This work was performed at Oak Ridge National Laboratory, operated by UT-Battelle for the~U.S.~Department of Energy under Contract~No.~DE-AC05-00OR22725.~The author acknowledges support from the laboratory directed research and development program.

\appendix

\section{Technical results}~In order for us to provide a more precise description of our semidefinite programs, we would need to introduce a few mathematical notations; some of which may be different from those used in the main text.~We let Alice's and Bob's complex Hilbert spaces be denoted by $\mathcal{A}$ and $\mathcal{B}$, respectively.~The set of linear operators, Hermitian operators and positive semidefinite operators acting on the composite Hilbert space are written as $\tn{L}(\mathcal{A}\otimes \mathcal{B})$, $\tn{Herm}(\mathcal{A}\otimes \mathcal{B})$ and $\tn{Pos}(\mathcal{A}\otimes \mathcal{B})$, respectively. Furthermore, we write $Q \succeq 0$ to indicate that $Q$ is positive semidefinite.~The set of density operators acting on Alice's and Bob's systems is defined as $\tn{D}(\mathcal{A} \otimes \mathcal{B}):=\{\rho \in \tn{Pos}(\mathcal{A}\otimes \mathcal{B}) : \Tr[ \rho ]=1\}$.~The set of separable operators, which is a closed convex cone, is denoted by $\tn{Sep}(\mathcal{A}:\mathcal{B})$.~Additionally, we would require the partial transpose operation, $T_\mathcal{B}=\mathbb{I}_{\tn{L}(\mathcal{A})} \otimes T$, which performs the transpose operation, $T$, on Bob's Hilbert space.~Accordingly, the set of positive partial transpose (PPT) operators is defined as $\tn{PPT}(\mathcal{A} : \mathcal{B}):=\{Q:T_\mathcal{B}(Q) \succeq 0, Q \in \tn{Pos}(\mathcal{A}\otimes \mathcal{B}) \}$.~Also, we denote a diagonal matrix by $Q=\tn{diag}[\lambda_1, \lambda_2, \lambda_3,\lambda_4]$.

\subsection{Accessible entanglement in separable states}~Let us first point out a key observation that explains why global measurements are more predictive than separable measurements for our choice of quantum input states.~Recall that the goal is to (locally) discriminate between two mixed separable states, namely,
 \[\rho_{0}=\frac{1}{8}\!\!\!\!\!\!\sum_{\substack{\bar{x},\bar{y}\\\tn{s.t.}f(\bar{x},\bar{y})=0}}\!\!\!\!\!\!\omega_{\bar{x}}\otimes \tau_{\bar{y}},\quad \rho_{1}=\frac{1}{8}\!\!\!\!\!\!\sum_{\substack{\bar{x},\bar{y}\\\tn{s.t.}f(\bar{x},\bar{y})=1}}\!\!\!\!\!\!\omega_{\bar{x}}\otimes\tau_{\bar{y}}. \] where $\omega_{\bar{x}}$ and $\tau_{\bar{y}}$ are defined as $\omega_{\bar{x}}=H^{x_1}\proj{x_2}H^{x_1}$ for $\bar{x}=x_1x_2 \in \{0,1\}^2$ and $\tau_{\bar{y}}=H^{y_1}\proj{y_2}H^{y_1}$ for $\bar{y}=y_1y_2 \in \{0,1\}^2$, respectively.~An important feature of these mixed states is that they can be simultaneously diagonalized in an entangled eigenbasis, whose eigenvectors are given by entangled states.~That is, using the standard Bell state definitions, i.e., $\ket{\Phi^\pm}=(\ket{00}\pm\ket{11})/\sqrt{2}$ and $\ket{\Psi^\pm}=(\ket{01}\pm\ket{10})/\sqrt{2}$, we have

 \be 
 \label{Eq7_mainresult:diag_basis} \rho_0=   \begin{pmatrix}
    \lambda^+ & 0 & 0 & 0 \\
    0 &\lambda^- & 0 & 0 \\
    0 & 0& \frac{1}{4}& 0 \\
    0 & 0 & 0 & \frac{1}{4} 
  \end{pmatrix}, \quad \rho_1=   \begin{pmatrix}
    \lambda^- & 0 & 0 & 0 \\
    0 & \lambda^+ & 0 & 0 \\
    0 & 0& \frac{1}{4}& 0 \\
    0 & 0 & 0 & \frac{1}{4}
  \end{pmatrix},
 \ee where the eigenvalues are $\lambda^\pm=(1\pm1/\sqrt{2})/4$, and the corresponding eigenvectors given as $\ket{\phi_1}=\sqrt{2\lambda^+}\ket{\Phi^-}+\sqrt{2\lambda^-}\ket{\Psi^+}$, $\ket{\phi_2}=\sqrt{2\lambda^-}\ket{\Phi^-}-\sqrt{2\lambda^+}\ket{\Psi^+}$, $\ket{\phi_3}=\ket{\Phi^+}$ and $\ket{\phi_4}=\ket{\Psi^-}$.~For example, we have $\lambda^+=\bra{\phi_1} \rho_0 \ket{\phi_1}=\bra{\phi_2} \rho_1 \ket{\phi_2}$.~From equation~(\ref{Eq7_mainresult:diag_basis}), we immediately see that a good guess for the optimal global measurement strategy is to first project the unknown state onto the subspace $\proj{\phi_1}+\proj{\phi_2}$, and then discriminate between the two orthogonal maximally entangled states $\ket{\phi_1}$ and $\ket{\phi_2}$ (i.e., to pick the maximum eigenvalue).~Indeed, this gives us a conclusive rate of $\lambda^+ + \lambda^-=1/2$, which agrees with our optimal solution found using semidefinite programming.~If we were to use separable measurements, then it is clear that some amount of mixing between the eigenvalues would occur, thus leading to a guessing value that is less than the maximum eigenvalue. 
\subsection{Optimal guessing probabilities}{\label{A_2}}~As mentioned in the main text, the bounds for general and PPT measurements can be analytically solved using convex optimization techniques, namely, semidefinite programming~\cite{SDP1996}.~More specifically, the idea is to find feasible solutions for the primal and dual programs, which provide lower and upper bounds on the optimal value, i.e., by virtue of the weak duality principle.~If the feasible solutions lead to values that coincide, then we say that the optimal solution for the semidefinite program is found.~That is, by the strong duality principle, the duality gap is zero.~In fact, the considered semidefinite programs have zero duality gaps.

Recall that the generic guessing probability defined in the state discrimination game for a fixed conclusive rate $\gamma^2$ is given as
\begin{eqnarray} \nonumber
G(\gamma)&:=&\frac{\Tr\left[\frac{1}{2}\rho_0\Pi_{a \oplus b =0}+\frac{1}{2}\rho_1\Pi_{a \oplus b =1} \right]}{\gamma^2}\\&=&\frac{\Pr\left[a\oplus b = f(\bar{x},\bar{y}) \right]}{\gamma^2}=\frac{1}{2}+\frac{S(\gamma)}{8},  \end{eqnarray}
where we used the measurement assignments $\Pi_{a\oplus b=0}=M_{0,0}+M_{1,1}$, $\Pi_{a\oplus b=1}=M_{0,1}+M_{0,1}$ and $\Pi_{\varnothing}=\mathds{1}-\Pi_{a\oplus b=0}-\Pi_{a\oplus b=1}$. Here, we remind that the measurement $\{M_{a,b}\}_{a,b}$ for all $a,b=0,1,\varnothing$ can be either a PPT measurement or a general global measurement, depending on which bound we want to solve. In the following, we will first show the computation for general global measurements.

\begin{thm}{\tn{\textbf{(Optimal guessing probability for general measurements).}}~The maximum probability of discriminating $\rho_0$ and $\rho_1$ using measurements $\{\Pi_0,\Pi_1,\Pi_\varnothing \} \in \tn{Pos}(\mathcal{A} \otimes \mathcal{B}) $ with a fixed conclusive rate of $\gamma^2 \in (0,1]$ is  }
\be \label{Supp_thm1}
\tn{max}\, G(\gamma) = \left\{\begin{array}{lll}
             \frac{1}{2}\left(1+\frac{1}{\gamma^22\sqrt{2}}\right)\quad &\tn{if}&\quad \gamma>\frac{1}{\sqrt{2}}\\
             \frac{1}{2}+\frac{1}{2\sqrt{2}} \quad &\tn{if}&\quad \gamma\leq\frac{1}{\sqrt{2}}
            \end{array}\right.,
\ee
\end{thm}
\begin{proof} The optimal solution is obtained if the feasible solutions for the primal and dual programs lead to a common optimization value.~To this end, the primal program for general measurements under the constraint that the conclusive rate is fixed to $\gamma^2 \in (0,1]$ is \\

\noindent 
\underline{{Primal program (general)}}
\begin{eqnarray*}\nonumber
\texttt{maximize}&:& \frac{1}{2}\Tr \left[\rho_0  \Pi_{a\oplus b=0}  + \rho_1 \Pi_{a\oplus b=1} \right], \\ \nonumber
\texttt{subject to}&:&  \Pi_{a\oplus b=0}+ \Pi_{a\oplus b=1}+ \Pi_{\varnothing} = \mathds{1}_{\mathcal{A} \otimes \mathcal{B}}\\  \nonumber
&&  \Tr\left[ (\omega_{\bar{x}}\otimes \tau_{\bar{y}}\Pi_{\varnothing}\right] =1-\gamma^2,\quad \forall~\bar{x},\bar{y}\\
&& \Pi_i \in \tn{Pos}(\mathcal{A} \otimes \mathcal{B}),\quad i=0,1,\varnothing,
\end{eqnarray*}
and the corresponding dual program is found to be \newline

\noindent
\underline{{Dual program (general)}}
\begin{eqnarray*}\nonumber
\texttt{minimize}&:& \Tr\left[Y \right] - (1-\gamma^2)\gamma\\ \nonumber
\texttt{subject to}&:& 2Y - \rho_i \succeq 0,\quad i=0,1\\  \nonumber
&& 4Y - \gamma\mathds{1}_{\tn{L}(\mathcal{A}\otimes \mathcal{B})} \succeq 0 \\ \nonumber
&& Y \in \tn{Herm}(\mathcal{A} \otimes \mathcal{B})\\
&& \gamma \in \mathbb{R}.
\end{eqnarray*}
For the region $0<\gamma \leq 1/\sqrt{2}$, we use the observations from the preceding section to construct a feasible solution for the primal program that is diagonal in the basis of equation~(\ref{Eq7_mainresult:diag_basis}), i.e., \begin{eqnarray*}
\tilde{\Pi}_{a\oplus b=0}&=&\tn{diag}\left[2\gamma^2, 0, 0 ,0 \right], \\ \tilde{\Pi}_{a\oplus b=1}&=&\tn{diag}\left[0, 2\gamma^2,0,0\right],\\ \tilde{\Pi}_{\varnothing}&=&\tn{diag}\left[1-2\gamma^2, 1-2\gamma^2,1,1 \right].
\end{eqnarray*} A direct computation of the primal objective function using this solution gives $\tn{max}\, \gamma^2{G}(\gamma) \geq \gamma^2\tilde{G}(\gamma) =\gamma^2\left(1+1/\sqrt{2}\right)/2$. A feasible solution for the dual in the same region is 
\begin{eqnarray*}
\tilde{Y}=\frac{1}{8}\left(1+\frac{1}{\sqrt{2}}\right)\mathbb{I}_{\tn{L}(\mathcal{A}\otimes \mathcal{B})},\quad \tilde{\gamma}=\frac{1}{2}\left(1+\frac{1}{\sqrt{2}}\right), 
\end{eqnarray*}which gives $\tn{max}\, \gamma^2{G}(\gamma) \leq \gamma^2\tilde{G}(\gamma) =\gamma^2\left(1+1/\sqrt{2}\right)/2$. Therefore, we arrive at the optimal solution (i.e., equation~(\ref{Supp_thm1})) for the region $0<\gamma\leq 1/\sqrt{2}$. 

For the other half of the region, $1/\sqrt{2} < \gamma \leq 1$,~a feasible solution for the primal program is 
\begin{eqnarray*}
\tilde{\Pi}_{a\oplus b=0}&=&\tn{diag}\left[1,\, \gamma^2-\frac{1}{2},\, \gamma^2-\frac{1}{2},0 \right], \\ \tilde{\Pi}_{a\oplus b=1}&=&\tn{diag}\left[0, 1, \,\,\gamma^2-\frac{1}{2},\gamma^2-\frac{1}{2} \right],\\ \tilde{\Pi}_{\varnothing}&=&\tn{diag}\left[0, 0,2(1-\gamma^2),2(1-\gamma^2) \right],
\end{eqnarray*} which leads to a lower bound of $\tn{max}\, \gamma^2{G}(\gamma) \geq \gamma^2\tilde{G}(\gamma) =(2\gamma^2+1/\sqrt{2})/4$.~A feasible solution for the dual program is 
\begin{eqnarray*}
\tilde{Y}=\begin{bmatrix}
    \mu_1 & 0 & 0 & -\mu_2 \\
    0 & \mu_1 & \mu_2 & 0 \\
    0 & \mu_2 & \mu_1 & 0 \\
    -\mu_2 & 0 & 0 & \mu_1
  \end{bmatrix},\quad \tilde{\gamma}=\frac{1}{2}, \end{eqnarray*}
  where $\mu_1=(2+1/\sqrt{2})/16$ and $\mu_2=(1+1/\sqrt{2})/8-\lambda_1$.~Plugging these into the dual objective gives $\tn{max}\, \gamma^2{G}(\gamma) \leq \gamma^2\tilde{G}(\gamma) =\left(2\gamma^2+1/\sqrt{2}\right)/4$.~Combining the lower and upper bounds, we thus get the other half of equation~(\ref{Supp_thm1}), that is,  $\tn{max}\, {G}(\gamma) =\left(1+1/(\gamma^2 2\sqrt{2})\right)/2$ for $\gamma > 1/\sqrt{2}$. \end{proof}

The upper bound for the local operations and shared randomness (LOSR) bound is computed using PPT measurements, which admit a concise mathematical characterization. Furthermore, under the assumption of two-qubit measurements, PPT measurements are necessarily separable measurements, since for any linear operator $Q\in \tn{L}(\mathcal{A}\otimes \mathcal{B})$, it is separable if and only if it is PPT~\cite{HorodeckiPPT}.~That is, the set of separable operators and the set of PPT operators are equivalent up to some constant factor for two-qubit positive operators.~Note that this is generally not the case if we consider higher dimension Hilbert spaces where PPT does not imply separability. In the below, we first show the optimal guessing probability assuming PPT measurements. 

\begin{thm}{\tn{\textbf{(Optimal guessing probability for PPT measurements).}}~The maximum probability of discriminating $\rho_0$ and $\rho_1$ using measurements $\{\Pi_0,\Pi_1,\Pi_\varnothing \} \in \tn{PPT}(\mathcal{A} \otimes \mathcal{B}) $ for any conclusive rate $\gamma^2 \in (0,1]$ is  }
\be \label{Supp_thm2}
\tn{max}\, G( \cdot|\tn{PPT}) =\frac{1}{2}+\frac{1}{4\sqrt{2}}.
\ee
\end{thm}
\begin{proof}~The primal program for separable/PPT measurements is given as \\

 \noindent 
\underline{{Primal program (separable/PPT)}}
\begin{eqnarray*}\nonumber
\texttt{maximize}&:& \frac{1}{2}\Tr \left[\rho_0  \Pi_{a\oplus b=0}  + \rho_1 \Pi_{a\oplus b=1} \right] \\ \nonumber
\texttt{subject to}&:&  \Pi_{a\oplus b=0}+ \Pi_{a\oplus b=1}+ \Pi_{\varnothing} = \mathds{1}_{\mathcal{A} \otimes \mathcal{B}}\\  \nonumber
&& \Tr\left[ (\omega_{\bar{x}}\otimes \tau_{\bar{y}}\Pi_{\varnothing}\right] =1-\gamma^2,\quad \forall~\bar{x},\bar{y}\\
&& \Pi_i \in \tn{PPT}(\mathcal{A} :\mathcal{B}),\quad i=0,1,\varnothing,
\end{eqnarray*}and the corresponding dual program is \newline

\noindent
\underline{{Dual program (separable/PPT)}}
\begin{eqnarray*}\nonumber
\texttt{minimize}&:& \Tr\left[Y \right] - (1-\gamma^2)\gamma\\ \nonumber
\texttt{subject to}&:& 2\left(Y -T_{\mathcal{B}}(Q_i)\right)- \rho_i  \succeq 0,\quad i=0,1\\  \nonumber
&& 4\left(Y-T_{\mathcal{B}}(Q_2)\right) - \gamma\mathds{1}_{\tn{L}(\mathcal{A}\otimes \mathcal{B})} \succeq 0 \\ \nonumber
&& Y \in \tn{Herm}(\mathcal{A} \otimes \mathcal{B})\\
&& Q_i \in \tn{Pos}(\mathcal{A} \otimes \mathcal{B}),\quad i=0,1,2 \\
&& \gamma \in \mathbb{R}.
\end{eqnarray*} Similarly, we construct feasible solutions for the primal and dual programs and show that their optimization values are identical. For the primal program, a feasible solution is 
\begin{eqnarray*}
\tilde{\Pi}_{a\oplus b=0}&=&\tn{diag}\left[\gamma^2, 0, \frac{\gamma^2}{2}, \frac{\gamma^2}{2}\right], \\ \tilde{\Pi}_{a\oplus b=1}&=&\tn{diag}\left[0,\gamma^2, \frac{\gamma^2}{2},  \frac{\gamma^2}{2}\right],\\ \tilde{\Pi}_{\varnothing}&=&\tn{diag}\left[1-\gamma^2, 1-\gamma^2,1-\gamma^2,1-\gamma^2 \right],
\end{eqnarray*} where each element is diagonal in the basis of equation~(\ref{Eq7_mainresult:diag_basis}).~Using this solution, we get $\tn{max}\, \gamma^2{G}(\gamma|\tn{Sep}) \geq \gamma^2 \tilde{G}(\cdot|\tn{Sep}) =\gamma^2\left(2+1/\sqrt{2}\right)/4$. For the dual program, a feasible solution is
\begin{eqnarray*}
\tilde{Y}=\frac{1}{8}\left(1+\frac{1}{2\sqrt{2}}\right)\mathds{1}_{\tn{L}(\mathcal{A}\otimes \mathcal{B})},\quad \tilde{\gamma}=\frac{1}{2}\left(1+\frac{1}{2\sqrt{2}}\right),  \\
Q_0=\frac{1}{8\sqrt{2}}\proj{\phi_2}, \quad Q_1=\frac{1}{8\sqrt{2}}\proj{\phi_1},\quad Q_2=0_{\tn{L}(\mathcal{A}\otimes \mathcal{B})},
\end{eqnarray*} which gives $\tn{max}\, \gamma^2{G}(\gamma|\tn{Sep}) \leq \gamma^2\tilde{G}(\cdot|\tn{Sep}) =\gamma^2\left(2+1/\sqrt{2}\right)/4$.~Therefore, after normalization, the obtained upper and lower bounds give equation~(\ref{Supp_thm2}).
\end{proof}

\subsection{Possible detection of all entangled two-qubit Werner states}\label{A_3} ~The detection of entangled two-qubit Werner states~\cite{Werner1989} can be shown by explicitly modeling the measurement operators in terms of two local measurements and a two-qubit Werner state $\phi_\xi:=\xi \proj{\Psi^-}+(1-\xi)\mathds{1}/4$ defined in two auxiliary systems $\mathcal{A}'$ and $\mathcal{B}'$. More specifically, the resulting measurements on systems $\mathcal{A}$ and $\mathcal{B}$ are given as 
\be
M_{a,b}=\Tr_{\mathcal{A}'\mathcal{B}'}\left[\phi_\xi M_{a,b}^+ \right],
\ee
where $M_{a,b}^+ \in \tn{Sep}(\mathcal{A}\otimes \mathcal{A}':\mathcal{B}\otimes \mathcal{B}')$. Indeed, the resulting measurements $\{M_{a,b}\}_{a,b}$ can only be entangled only if the underlying Werner state is entangled.~We compute (upper) quantum bounds for this choice of modeling using the following semidefinite program assuming $\gamma^2=1/4$.\\

\noindent\underline{{Primal program (Werner states with fixed $\xi$)}}
\begin{eqnarray*}\nonumber
\texttt{maximize}&:& \frac{1}{2}\Tr \left[(\rho_0 \otimes \phi_\xi)  \Pi_{a\oplus b=0}  + (\rho_1\otimes \phi_\xi) \Pi_{a\oplus b=1} \right] \\ \nonumber
\texttt{subject to}&:&  \Pi_{a\oplus b=0}+ \Pi_{a\oplus b=1}+ \Pi_{\varnothing} = \mathds{1}_{\mathcal{A} \otimes \mathcal{A}' \otimes \mathcal{B} \otimes \mathcal{B}'}\\  \nonumber
&&\frac{1}{2}\tr\left[ (\rho_0+\rho_1)\otimes \phi_\xi  \Pi_{\varnothing}\right] =3/4,\quad i=0,1\\
&& \Pi_i \in \tn{PPT}(\mathcal{A}\otimes \mathcal{A}':\mathcal{B}\otimes \mathcal{B}'),\quad i=0,1,\varnothing.
\end{eqnarray*}
\begin{figure}[t]
  \includegraphics[width=90mm]{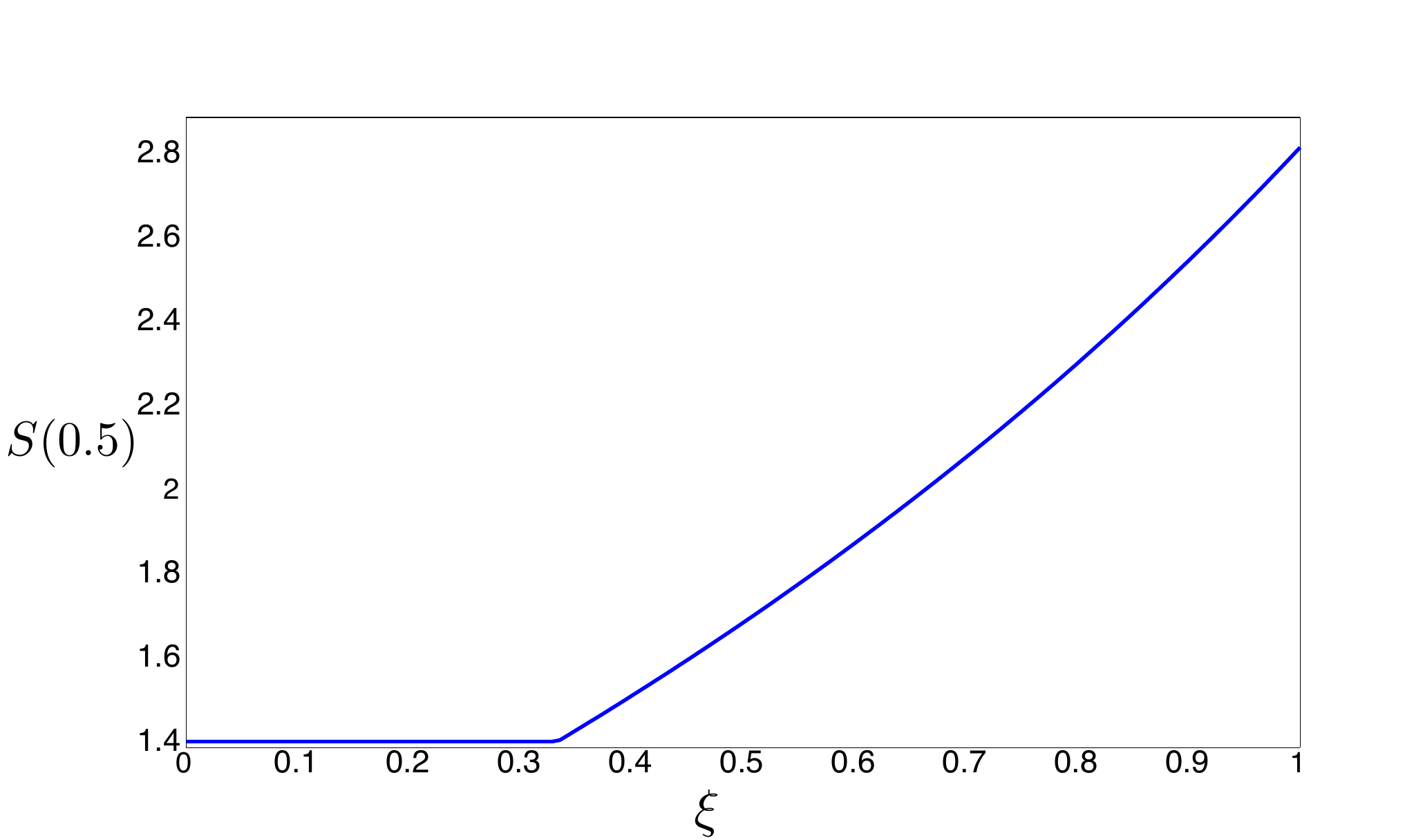}
  \caption{{\textbf{Quantum violation vs maximally entangled fraction.}~Here, we see that quantum violations are obtained only for $\xi > 1/3$, which means that the proposed semi-quantum CHSH inequality could be capable of detecting all entangled two-qubit Werner states. }~
}    \label{fig4}
\end{figure}
Here, the optimization of the semidefinite program is taken over measurements separable with respect to the bi-partition $\mathcal{A}\otimes \mathcal{A}':\mathcal{B}\otimes \mathcal{B}'$.

\bibliographystyle{unsrt}

\begin{thebibliography}{99}
\bibitem{footnote0}
That is, events happening in one space-time region are independent of events in space-like separated regions.
\bibitem{footnote1}
The measurement efficiency of a measurement device is defined as the probability of obtaining a conclusive outcome whenever the device receives an input.
\bibitem{footnote2}
More precisely, the eigenvectors of $\mathbbm{X}$ are $\{(\ket{0}\pm\ket{1})/\sqrt{2}\}$ and the ones of $\mathbbm{Y}$ are $\{(\ket{0}\pm i\ket{1})/\sqrt{2}\}$.
\bibitem{footnote3}
The standard Bell basis is defined as $\ket{\Phi^\pm}:=(\ket{00}\pm\ket{11})/\sqrt{2}$ and $\ket{\Psi^\pm}:=(\ket{01}\pm\ket{10})/\sqrt{2}$.
\end{thebibliography}

\end{document}